%% file: main.tex
\documentclass[runningheads]{llncs}
\usepackage{graphicx} 
\usepackage{adjustbox}
\usepackage[inline]{enumitem}
\usepackage{blindtext}
\usepackage{orcidlink}
\usepackage{caption}
\usepackage{subcaption}
\usepackage{listings}
\usepackage{lstautogobble}
\usepackage{tikz}
\usepackage{amsmath}
\usepackage{amssymb}
\usepackage{stmaryrd}
\usepackage{cleveref}
\usepackage{todonotes}
\usepackage{wrapfig}
\usepackage{mathtools}
\usepackage{mleftright}
\usepackage{ifthen}

\newboolean{fullversion}
\setboolean{fullversion}{true}

\usetikzlibrary{calc,positioning,fit,backgrounds,matrix}

\definecolor{bluekeywords}{rgb}{0.13, 0.13, 1}
\definecolor{greentypes}{rgb}{0, 0.5, 0}
\definecolor{orangecomments}{rgb}{1, 0.5, 0.1}
\definecolor{redstrings}{RGB}{171, 114, 2}
\definecolor{graynumbers}{rgb}{0.5, 0.5, 0.5}
\definecolor{goldcomments}{rgb}{0.6, 0.4, 0.08}

\definecolor{implemented}{rgb}{0.67, 0.9, 0.93}
\colorlet{existing}{lightgray}

\lstdefinelanguage{Lola}{
  keywords=[0]{input, output, trigger, constant, import, spawn, eval, close, with, when, schedule},
  keywordstyle=[0]\bfseries\color{bluekeywords},
  keywords=[1]{if, then, else, aggregate, defaults, offset, by, or, to, sin, cos, abs, hold, over, using, min, max, in},
  keywords=[2]{Variable, String, Int, Int8, Int64, UInt, UInt8, UInt64, Bool, Float32, Float64, Float, @1Hz, @5Hz, @10Hz, @100mHz, @1kHz, @1min},
  keywordstyle=[2]\color{greentypes},
  moredelim=[is][\textcolor{greentypes}]{|}{|},
  moredelim=[s][\color{teal}\bfseries]{\#[}{]},
  moredelim=[s][\color{teal}\bfseries]{\#![}{]},
  sensitive=false,
  comment=[l]{//},
  morecomment=[s]{/*}{*/},
  morestring=[b]',
  morestring=[b]"
}

\input{macros.tex}

\begin{document}
\title{Active Monitoring with RTLola:\\A Specification-Guided Scheduling Approach}
\titlerunning{Active Monitoring with RTLola}
\author{
  Jan Baumeister\orcidlink{0000-0002-8891-7483}
  \and Bernd Finkbeiner\orcidlink{0000-0002-4280-8441}
  \and Frederik Scheerer\orcidlink{0009-0007-8115-0359}
}
\authorrunning{Baumeister et~al.}

\institute{CISPA Helmholtz Center for Information Security,\\Saarbrücken, Germany\newline
\email{\{jan.baumeister, finkbeiner, frederik.scheerer\}@cispa.de}}
 \maketitle
 \begin{abstract}

  Stream-based monitoring is a well-established runtime verification approach which relates input streams, representing sensor readings from the monitored system, with output streams that capture filtered or aggregated results.
  In such approaches, the monitor is a passive external component that continuously receives sensor data from the system under observation.
  This setup assumes that the system dictates what data is sent and when, regardless of the monitor's current needs.
  However, in many applications -- particularly in resource-constrained environments like autonomous aircraft, where energy, size, or weight are limited -- this can lead to inefficient use of communication resources.
  We propose making the monitor an active component that decides, based on its current internal state, which sensors to query and how often.
  This behavior is driven by scheduling annotations in the specification, which guide the dynamic allocation of bandwidth towards the most relevant data, thereby improving monitoring efficiency.
  We demonstrate our approach using the stream-based specification language RTLola and asses the performance by monitoring a specification from the aerospace domain.
  With equal bandwidth usage, our approach detects specification violations significantly sooner than monitors sampling all inputs at a fixed frequency.

 \keywords{Stream-based Monitoring \and Constraint-Based Scheduling \and Real-time Properties}
 \end{abstract}
\input{intro.tex}
\input{rtlola.tex}
\input{scheduling.tex}
\input{implementation.tex}
\input{evaluation.tex}
\input{conclusion.tex}

\subsubsection{Acknowledgments.} This work was partially supported by the German Research Foundation (DFG) as part of TRR 248 (No.~389792660) and by the European Research Council (ERC) Grant HYPER (No.~101055412).

\bibliographystyle{splncs04}
\bibliography{bibliography.bib}

\ifthenelse{\boolean{fullversion}}{
\clearpage
\appendix
\input{appendix}
}{
}

\end{document}

%% file: macros.tex
\newcommand{\rtlola}{RTLola\xspace}
\newcommand{\stream}{\text{Stream}}
\newcommand{\streammap}{\text{StreamMap}}
\renewcommand{\time}{\text{Time}}
\newcommand{\timemap}{\text{TimeMap}}
\newcommand{\World}{\mathbb{W}}
\newcommand{\world}{\omega}
\newcommand{\RR}{\mathbb{R}}
\newcommand{\VV}{\mathbb{V}}
\newcommand{\NN}{\mathbb{N}}
\newcommand{\BB}{\mathbb{B}}

\newcommand{\inputref}{\mathtt{ID}^\uparrow}
\newcommand{\outputref}{\mathtt{ID}^\downarrow}
\newcommand{\sid}{\mathit{sid}}
\newcommand{\condition}{\text{Cond}}
\newcommand{\Sref}{\mathtt{ID}}
\newcommand{\sref}{\mathit{sref}}

\newcommand{\schedulingDecision}{\mathbb{S}}
\newcommand{\scheduleYes}{\text{Y}}
\newcommand{\scheduleMaybe}{\text{M}}
\newcommand{\scheduleNo}{\text{N}}
\newcommand{\Tasks}{\text{Tasks}}
\newcommand{\task}{\tau}
\newcommand{\tasks}{\mathcal{T}}

\newcommand{\Schedule}{S}

\newcommand{\overdue}{\mathit{overdue}}

\newcommand{\IntervalSchedule}{\text{Deadline}}
\newcommand{\PrioritySchedule}{\text{Priority}}
\newcommand{\staticSchedule}{\xi}
\newcommand{\dynamicSchedule}{\psi}
\newcommand{\taskSat}[3]{#2 \models_{#3} #1}
\newcommand{\taskUnSat}[3]{#2 \not\models_{#3} #1}
\newcommand{\lastPrio}{\mathit{Prio}}
\DeclareMathOperator*{\argmax}{arg\,max}
\newcommand{\prefix}{\text{prefix}}
\newcommand{\SA}{\text{SA}}
\newcommand{\validTasks}{\mathit{validTasks}}
\newcommand{\simpleBound}[1]{\text{InputEventBound}_{#1}}
\newcommand{\inputValues}{\text{InputValues}}
\newcommand{\takeNextEvent}{\mathit{takeEvent}}

\newcommand{\set}[1]{\left\{#1\right\}}
\newcommand{\setCond}[2]{\mleft\{#1\;\middle|\;#2\mright\}}
\newcommand{\evalTo}[2]{\Downarrow^{#1}_{#2}}
\newcommand{\semantics}[1]{\llbracket#1\rrbracket}
\newcommand{\powerset}[1]{\mathcal{P}\left(#1\right)}
\renewcommand{\iff}{\mathbin{\mathrm{iff}}}

%% file: intro.tex
\section{Introduction}

Cyber-physical systems are increasingly prevalent, and many now operate fully autonomously in complex, real-world environments.
These systems are often deployed in safety-critical domains, such as autonomous vehicles or drones, where incorrect behavior can lead to catastrophic outcomes.
Runtime monitoring is a well-established technique for checking the system's behavior at runtime against a formal specification and, therefore, ensuring its correct behavior~\cite{DBLP:series/natosec/FalconeHR13,DBLP:conf/vstte/HavelundG05,DBLP:journals/jlp/LeuckerS09}.

A prominent class of runtime monitoring frameworks is stream-based monitoring~\cite{baumeister2024tutorial,kallwies2022tessla,gorostiaga2018striver}.
There, the system continuously supplies a monitor with data about its current state via input streams.
Output streams compute new values by aggregating and filtering the inputs, often leveraging temporal operators to express rich, time-dependent specifications.
This makes stream-based monitoring particularly suitable for complex safety requirements of cyber-physical systems.

However, in modern autonomous systems, monitors must process data from a diverse set of sources in real-time.
For example, an autonomous vehicle or drone might rely on GPS, camera feeds, and many other sensors to assess its environment and maintain safety.
Here, the monitor is a passive component, consuming the updates about the current state of the system under observation.
A key challenge in such settings is limited bandwidth: the connection between sensors and the monitor cannot support arbitrarily high-frequency updates from all sources due to factors like energy consumption, physical size, or weight constraints.
Consequently, it is not feasible to transmit all sensor data at its highest possible frequency.
However, reducing the sampling frequency leads to increased latency in detecting critical events and dangerous situations, undermining the effectiveness of the monitor.

We propose that, since the monitor has a comprehensive view of all received data, it has a better understanding of which data is currently needed than the different sources individually.
For instance, consider a drone equipped with both an altitude sensor and a camera.
When the drone is landed, altitude readings are irrelevant, while the ground-level camera might be more important for monitoring the surrounding ground.
Conversely, at higher altitudes, the camera data becomes less informative, whereas altitude measurements become more significant as the drone is approaching predefined upper altitude limits.
We present an active monitoring approach in which the system adaptively queries the sensors based on the drone's current altitude.
This enables dynamic prioritization of sensor data, emphasizing ground-level visual input when near the ground and focusing on altitude measurements when high up in the sky -- therefore making more effective use of the limited bandwidth between the sensors and the monitor.

Such an approach is especially interesting for applications that are designed for actively requesting individual sensors' data.
A prominent example is the OBD-II interface present in all modern vehicles, which must be explicitly queried for individual sensors of the car, instead of receiving a stream of new sensor data on its own.

Given the availability of well-established monitoring tools~\cite{baumeister2024tutorial,kallwies2022tessla,DBLP:conf/fm/PerezGD24,DBLP:conf/rv/RozierS17}, we chose not to develop a new monitoring framework from scratch.
Instead, our approach builds on existing monitoring tools by introducing a scheduling component that interfaces between the monitor and the sensors.
This scheduler observes the current monitoring state and determines which sensors to query next.
Once the selected data has been acquired, the scheduler forwards it to the underlying monitor to initiate the next monitoring cycle.
This architecture allows our method to remain independent of the monitor implementation and be compatible with a wide range of existing monitoring infrastructures.

We demonstrate our approach through an implementation based on the stream-based monitoring language \rtlola~\cite{baumeister2024tutorial}, which has previously been successfully applied to the monitoring of unmanned aircraft~\cite{DBLP:conf/cav/BaumeisterFKLMST24,DBLP:conf/cav/BaumeisterFSST20} or for increasing trust in automatic decision and prediction systems~\cite{baumeister_stream-based_2025}.
We extend \rtlola by introducing scheduling annotations that can be attached to individual streams in the specification.
These annotations guide the scheduling process by indicating which inputs should be prioritized under certain conditions of the monitor.
The annotated specification is then automatically transformed into a "regular" \rtlola specification compatible with all existing \rtlola implementations~\cite{BCFS25,DBLP:conf/cav/FaymonvilleFSSS19,DBLP:journals/tecs/BaumeisterFST19}.
The schedule is embedded into this transformed specification through additional output streams, which are then interpreted by the scheduler to determine which inputs should be queried in the next cycle.

We evaluated our approach using data obtained through the Microsoft AirSim simulator.
Our results demonstrate that the proposed method can significantly reduce bandwidth consumption without compromising monitoring quality, or, on the other hand, detect violations earlier compared to a fixed-frequency approach, while utilizing the same amount of bandwidth.

\paragraph{Contribution.}
To summarize, we make the following contributions:
\begin{itemize}
    \item We define a formal semantics for stream-based specification languages that incorporate scheduling of streams,
    \item We extend \rtlola with an annotation mechanism to express scheduling constraints directly within the specification,
    \item We present a translation from annotated specifications to regular ones, which integrate the scheduling information as additional streams, and
    \item We implement our approach on top of \rtlola and evaluate it in an online setting using simulated drone data.
\end{itemize}

The remainder of this paper is structured as follows:
In \Cref{sec:rtlola}, we provide a background on the \rtlola monitoring language and define formal semantics for stream-based specification languages.
In \Cref{sec:scheduling_semantics}, we extend these semantics to incorporate the scheduling mechanism.
\Cref{sec:rtlola_scheduling} introduces our active monitoring framework and discusses challenges specific to integrating scheduling into \rtlola.
Finally, we present our evaluation in \Cref{sec:evaluation}.

\subsection{Related Work}

Many monitoring frameworks assume that every change in the system is observable by the monitor.
However, this assumption does not hold in bandwidth-constrained environments, where the monitor must instead periodically sample the system state.
Bonakdarpour et al.~\cite{DBLP:conf/fm/BonakdarpourNF11} propose a time-driven monitoring approach that samples the system at a fixed frequency to reduce overhead.
This frequency is statically determined to guarantee that no violations are missed, but it treats all inputs equally throughout execution.
Navabpour et al.~\cite{DBLP:conf/rv/NavabpourBF12} extend this work by making the sampling path-aware, allowing the frequency to vary over time. 
Stoller et al.~\cite{DBLP:conf/rv/StollerBSGHSZ11} also aim to reduce monitoring overhead by sampling the system periodically and use Hidden Markov Models to estimate unobserved states between samples.
In contrast to the previous approaches which sample the system, our setting allows explicitly querying individual sensor, giving us the potential to prioritze inputs according to their importance.
Huang et al.~\cite{DBLP:journals/sttt/HuangSCDGSSZ12} take a different approach by dynamically adjusting the amount of monitoring to stay within a user-specified target overhead.
Instead of sampling, they temporarily disable the monitoring of certain events when the limit is exceeded otherwise.

Our approach is designed for the stream-based specification language \rtlola~\cite{baumeister2024tutorial}, a successor of the synchronous stream-based specification language Lola~\cite{d2005lola}.
\rtlola has been successfully applied to the monitoring of cyber-physical systems, such as unmanned aircraft~\cite{DBLP:conf/cav/BaumeisterFKLMST24,DBLP:conf/cav/BaumeisterFSST20}.
However, our approach is also transferable to other asynchronous stream-based languages, including Tessla~\cite{kallwies2022tessla} and Striver~\cite{gorostiaga2018striver}.
While there exist stream-based approaches that utilize annotations in specifications for correctness guarantees~\cite{DBLP:journals/sttt/BaumeisterDFS23,DBLP:conf/fmcad/HagenT08}, we use annotations to represent timing constraints.

Another relevant area of research is scheduling in real-time systems, assigning tasks to processors -- a topic comprehensively overviewed by Buttazzo~\cite{buttazzo2011hard} -- and various approaches also consider bandwidth constraints (e.g., \cite{DBLP:conf/IEEEpact/XuWY10,DBLP:journals/corr/abs-2010-16058,DBLP:conf/sies/AfsharBBN15,DBLP:phd/dnb/Agrawal19}).
In this paper, we adapted these concepts to stream-based languages, by defining the scheduling semantics in the context of stream-based monitoring and integrating them into the specification language.
While scheduling of real-time systems defines dependencies between tasks, stream-based settings provide a different kind of dependency constraints, relating jointly evaluated tasks.
Our approach is closest to non-preemptive scheduling with dynamic priorities and hard aperiodic deadlines.

%% file: rtlola.tex
\section{RTLola}
\label{sec:rtlola}

Stream-based monitors operate over a set of \emph{streams}, each representing an infinite sequence of values.
\rtlola specifications define these streams through stream-equations, which describe how the values in the sequences are computed.
We distinguish between \emph{input streams}, which are populated with data from the monitored system, and \emph{output streams}, which compute new values by aggregating and filtering the inputs.

Consider the following example of an \rtlola specification:
\begin{lstlisting}
input alt : Float64
output alt_diff
    eval |@alt| with abs(alt - alt.offset(by:-1).defaults(to: 0.0))
trigger alt_diff > 10.0
\end{lstlisting}
The specification defines one input stream \lstinline!alt! which is automatically populated with new readings from the drone's altitude sensor as they arrive at the monitor.
Next, the specification defines the output stream \lstinline!alt_diff!, which computes the absolute difference between two consecutive altitude readings.
This is achieved using the offset operator, which allows access to past stream values.
Because such previous values are unavailable at the startup of the monitor, a default value must be provided that is used instead in this case.
Finally, a trigger is defined to check whether the altitude difference exceeds 10.
If the expression evaluates to true, the trigger activates, indicating a violation of the specification.

\subsection{Types}

Each stream in \rtlola has two associated types: a value type and a pacing type.
The \emph{value type} specifies the kind of data contained in the stream, such as integers, floating-point numbers, or booleans.
The \emph{pacing type} defines \emph{when} new values of a stream are computed and is indicated using the \lstinline!|@|!-symbol following the \lstinline!eval! keyword.
In this paper, all output streams are \emph{event-driven}, meaning they are evaluated whenever new inputs arrive.
For example, the \lstinline!alt_diff! stream always computes a new value whenever the \lstinline!alt! stream receives a new value, as indicated by its pacing type.

Next, consider the following extension of the previous specification:
\begin{lstlisting}
output num_high_alt
	eval |@alt| when alt > 20.0
          with num_high_alt.offset(by:-1).defaults(to: 0) + 1
\end{lstlisting}
This stream is evaluated whenever the \lstinline!alt! input stream receives a new value.
Here, the evaluation is further conditioned by the dynamic filter condition provided after the \lstinline!when! keyword: A new value is computed only if the condition evaluates to true.
In this case, the stream counts how often a high-altitude reading occurs.
\rtlola also allows multiple \lstinline!eval! clauses per stream.
Clauses are checked from top to bottom, and the stream is evaluated using the \lstinline!with! expression of the first clause whose \lstinline!when! condition is satisfied.

For a more detailed explanation of the \rtlola monitoring language -- including features such as time-driven streams, aggregations, and parametrization -- we refer the reader to the \rtlola Tutorial~\cite{baumeister2024tutorial}.

\subsection{Semantics}

We use the following semantics for stream-based monitors.
Although the formulation is applied to \rtlola, the semantics are transferable to accommodate other stream-based specification languages.
We define the semantics via an \emph{evaluation model} $\world \in \World$, which assigns stream references $\Sref$, consiting of input stream references $\inputref$ and output stream references $\outputref$, to a timed series of values.
We write $\world(t)$ to refer to the real-time timestamp at discrete timestep $t \in \time$, and $\world(\sid)(t)$ to refer to the value of the stream $\sid$ at time $t$.
If this stream did not calculate a new value at that time because of its pacing or filter conditions, $\bot$ is returned instead.

The semantics ensure that each output stream value in the evaluation model is correctly computed according to the specification's defining equations.
Formally, given an \rtlola specification $\varphi$, we define its semantics as the set:
\begin{align*}
	\semantics{\varphi} &= \big\{\world \in \World \;\mid\; \forall \sid \in \outputref. \forall t \in \time. \varphi(\sid) \evalTo{t}{w} \world(\sid)(t)\\
	&\hspace{2.1cm} \land \forall t \in \time. \world(t) < \world(t+1) \big\},
\end{align*}
In the formula, $\varphi(\sid) \evalTo{t}{w} v$ denotes that the defining stream equations $\varphi(\sid)$ of stream $\sid$ evaluate to $v$ at time $t$.
This evaluation yields the result of the \lstinline!with!-expression of the first \lstinline!eval! clause whose pacing and \lstinline!when! condition is satisfied at time $t$.
If no clause is evaluated, the result is $\bot$.


An evaluation model is considered valid with respect to a specification if, at every time step, all output stream values are correctly computed according to the defining stream equations and the time map is strictly monotonically increasing.

An \rtlola specification $\varphi$ is considered \emph{well-defined}, if for every possible input trace $I \in \time \rightarrow \inputValues$ with $\inputValues: \inputref \rightarrow \VV_\bot$, there exists a unique evaluation model $\world \in \semantics{\varphi}$ with $\forall t \in \time. \forall i \in \inputref. \world(i)(t) = I(t)(i)$.

%% file: scheduling.tex
\section{Scheduled Monitor Semantics}
\label{sec:scheduling_semantics}

This section introduces the general concept for stream-based scheduling, where a scheduler must dynamically adapt to the current state of the monitor.
We describe a valid schedule with a set of static constraints and later determine whether a scheduler satisfies them.
We start by defining dynamic schedule constraints to decide at each time point if the evaluation violates the constraint.
Then, we describe how to transform a static schedule into its dynamic counterpart -- the transformation later implemented by the scheduling component.

In our setting, the scheduler determines at each time step which tasks are evaluated in the next step.
The set of possible tasks defines the space of scheduling decisions available to the scheduler at each time step.
\begin{definition}[Task]
For each $\task \in \Tasks$ there exists:
\begin{enumerate*}
	\item a predicate $\taskSat{\task}{\world}{t}$ to dermine if the evaluation model $\world\in\World$ satisfies the task $\task$ at time $t$,
	\item a partial order $\preceq$ representing dependencies between tasks, and
	\item a predicate $iv  \models \tasks$ to determine if an input $iv: \inputValues$ reflects the set of tasks $\tasks$.
\end{enumerate*}
\end{definition}

To illustrate different choices for the task space, consider two examples.
\begin{example}[Individual Streams]
	For this example, consider a scheduler that selects individual streams during evaluation.
    Then $\Tasks = \Sref$ consisting of all stream's references and
	\begin{align*}
		\taskSat{\sref}{\world}{t} &\quad\iff\quad \world(\sref)(t) \ne \bot\\
        \task_1 \preceq \task_2\;&\quad\iff\quad \task_1 = \task_2\\
		iv \models \tasks &\quad\iff\quad \forall i \in \inputref. iv(i) \ne \bot \Leftrightarrow \exists \sref \in \tasks. \sref = i.
	\end{align*}
In this case, tasks are satisfied at time $t$ whenever their stream receives a value, i.e., is not $\bot$ at time $t$.
There exist no dependencies between tasks, and an input reflects a set of tasks if each input receives a new value iff it is contained in the set of tasks.
\end{example}
\begin{example}[Stream Sets]
Alternatively, consider a scheduler selecting groups of streams to be jointly evaluated.
We then set $\Tasks = \powerset{\Sref}$ and
\begin{align*}
\taskSat{\task}{\world}{t} &\quad\iff\quad \forall \sref \in \task. \world(\sref)(t) \ne \bot\\
\task_1 \preceq \task_2 &\quad\iff\quad \task_1 \subseteq \task_2\\
iv \models \tasks &\quad\iff\quad \forall i \in \inputref. iv(i) \ne \bot \Leftrightarrow \exists \task \in \tasks. i \in \task.\nonumber
\end{align*}
In this representation, tasks have dependencies through a subset relation, and an input reflects a set of tasks if at least one task updates each input.
\end{example}

\subsection{Dynamic Schedule Constraints}

Next, we introduce the concept of dynamic scheduling constraints, which specify how each task is expected to be evaluated:
\begin{definition}[Dynamic Schedule Constraint]
A \emph{dynamic schedule constraint} is a function that maps the current monitor state $(\World, \time)$ to a scheduling decision over tasks:
\begin{align*}
	\Schedule : (\World \times \time) \rightarrow \Tasks \rightarrow \schedulingDecision \hspace{2cm} \schedulingDecision &= \set{\scheduleYes, \scheduleMaybe, \scheduleNo}
\end{align*}
For each task, the constraint assignes whether the task must ($\scheduleYes$), may ($\scheduleMaybe$), or must not ($\scheduleNo$) be evaluated at that time.
\end{definition}

The semantics of stream-based languages ensure that all stream values are computed in accordance with their defining stream equations.
We extend this semantics to account for scheduling constraints, requiring the evaluation model to also satisfy the tasks according to a constraint.
In addition, we introduce a \emph{bandwidth bound} $B$, which encodes the bandwidth limitations between sensors and the monitor.
Formally, given a specification $\varphi$, a schedule constraint $\psi \in S$, and a bandwidth constraint $B$, the \emph{scheduled semantics} is defined as:
\begin{align*}
	\semantics{(\varphi, \psi, B)} = \big\{\world \in \World \;\mid\; \world \in \semantics{\varphi} \land \world \in \semantics{\psi} \land \forall t \in \time. B(\world, t)\big\}.
\end{align*}
Intuitively, an evaluation model is valid if it:
\begin{enumerate*}
	\item correctly computes all stream values according to $\varphi$,
	\item adheres to the scheduling decisions made by $\psi$, and
	\item respects the bandwidth constraint $B$ at all times.
\end{enumerate*}

For an evaluation model to satisfy a constraint $\psi\in S$, the satisfaction or non-satisfaction of tasks must always align with $\psi$:
\[
	\semantics{\psi} = \left\{\world \in \World \;\mid\; \forall t \in \time. \forall \task \in \Tasks. \begin{cases}
		\taskSat{\task}{\world}{t + 1} &\text{if } \psi(\world, t)(\task) = \scheduleYes\\
		\top &\text{if } \psi(\world, t)(\task) = \scheduleMaybe\\
		\taskUnSat{\task}{\world}{t + 1} &\text{if } \psi(\world, t)(\task) = \scheduleNo
	\end{cases}\right\}.
\]
If the schedule constraints assign $\scheduleYes$ to $\task$ at time $t$, the evaluation model must satisfy that task at time $t+1$.
If the decision is $\scheduleNo$, it must not satisfy that task, while for $\scheduleMaybe$, both outcomes are permitted.

Last, we define the bandwidth constraints imposed by the communication between sensors and the monitor.
These constraints are formalized as a predicate $B: (\World \times \time) \rightarrow \BB$, which determines whether the inputs received by the monitor at time $t$ conform to the bandwidth limitations.
In this paper, we consider a simple constraint model $\simpleBound{b}$ that limits the number of input streams that can receive a value at the same time to a fixed threshold $b$:
\[
	\simpleBound{b}(\world, t) = |\setCond{i \in \inputref}{\world(i)(t) \ne \bot}| \le b.
\]
Other constraints could account for the varying bit widths of individual input types or enforce protocols that require specific combinations of inputs.

\subsection{Static Scheduling Constraints}
\label{sec:static_schedule}

This section introduces three \emph{static schedule constraints} -- deadlines, priorities, and their combination -- and presents a translation deriving their dynamic counterpart.
Static schedule constraints contain \emph{conditions} $\condition$, which can be evaluated to a boolean under a given evaluation state $\evalTo{\world}{t}: \condition \rightarrow \BB$, to constrain a task differently at runtime.
In the following paragraphs, we describe these static schedules individually.

\subsubsection{Deadline}
The first static schedule constraints describe deadlines.
The constraints assign upper bounds to tasks, which indicate that a task should not be evaluated later than its deadline.
It is defined as
\[
	\IntervalSchedule : \Tasks \rightarrow \powerset{\condition \times \RR}
\]
and assigns each task to a set of pairs, each consisting of conditions and a corresponding deadline.
Given such a static schedule constraint $\staticSchedule \in \IntervalSchedule$, we derive a dynamic schedule $\psi \in \Schedule$ as follows:
\begin{align*}
	\dynamicSchedule(\world, t)(\task) = \begin{cases}
		\scheduleYes & \text{if } \exists (c, dl) \in \staticSchedule(\task) \land \exists t' < t. c \evalTo{\world}{t'} \top\\
		& \phantom{\text{if }} \land \forall t'' \in (t', t]. \taskUnSat{\task}{\world}{t''}\\
		& \phantom{\text{if }} \land \world(t+2) > \world(t') + dl\\
		\scheduleMaybe & \text{otherwise}
	\end{cases}
\end{align*}
Intuitively, assuming tasks represent individual streams, a stream must be evaluated at time $t+1$ (i.e., $\dynamicSchedule(\world, t)(\task) = \scheduleYes$) if there exists a condition in the static schedule that was satisfied at an earlier time $t'$ and no new value has been produced for that stream since.
Then, $t+1$ is the last chance to produce a value for that stream before violating its deadline.
Otherwise, the scheduling decision defaults to $\scheduleMaybe$, allowing the stream's evaluation to be postponed.


\subsubsection{Priority}

We may want to fully utilize the available bandwidth without manually assigning explicit deadlines. 
To support this use case, we introduce the \emph{priority}-based static schedule constraint.
In this case, each task is assigned a priority depending on conditions, dynamically determining its importance:
\[
	\PrioritySchedule : \Tasks \rightarrow \powerset{\condition \times \NN}.
\]
At each time step, the scheduler must select the streams with the highest priority. 
Given a static schedule constraint $\staticSchedule \in \PrioritySchedule$, we determine the current priority of a task with $\lastPrio_\staticSchedule$:
\begin{align*}
	\lastPrio_\staticSchedule &: \Tasks \times \World \times \time \rightarrow (\condition \times \NN)_\bot\\
	\lastPrio_\staticSchedule(\task, \world, t) &= \begin{dcases}
		\argmax_{(c,p) \in S} (\max \setCond{t'}{t' \le t \land c \evalTo{\world}{t'} \top}) & \text{if $(c,p)$ exists}\\
		\bot & \text{otherwise}
	\end{dcases}
\end{align*}
$\lastPrio_\staticSchedule$ returns the most recent priority assignment of a task (if any) whose condition evaluates to true in $\world$ at some time point $t' \le t$.
Using this function, we translate a static schedule $\staticSchedule \in \PrioritySchedule$ into a dynamic schedule $\dynamicSchedule \in \Schedule$ with
\begin{align*}
	\dynamicSchedule(\world, t)(\task) = \begin{cases}
	\scheduleYes & \text{if } (c_1, p_1) = \lastPrio_{\staticSchedule}(\task, \world, t)\\
	& \phantom{if } \land \exists \task'. (c_2, p_2) = \lastPrio_{\staticSchedule}(\task', \world, t)\\
	& \phantom{if } \land \taskSat{\task'}{\world}{t+1} \land p_1 > p_2\\
	& \phantom{if } \land \neg\exists \task'' \succeq \task. \taskSat{\task''}{\world}{t+1}\\
	\scheduleMaybe & \text{otherwise}
	\end{cases}
\end{align*}
With this definition, a task must be selected if another, lower-priority task is selected for evaluation at the next step.
This restriction, however, does not hold if the task is part of a larger, high-priority task, indicated by the dependency relation.


\subsubsection{Deadline and Priority}

The priority schedule allows specifying the relative importance of tasks without explicitly reasoning about individual deadlines.
However, this can lead to starvation: a lower-priority task may never be evaluated if higher-priority tasks continuously occupy all available bandwidth.
To address this issue, we propose a combined schedule that merges priorities with deadlines.
By default, the schedule behaves like the priority schedule.
However, each $\task$ is assigned a deadline $dl_{\task}$, which defines the maximum duration it may remain unevaluated before it is considered \emph{overdue}:
\begin{align*}
	&\overdue: \Tasks \times \World \times \time \rightarrow \BB\\
	&\overdue(\task, \world, t) = \exists \task' \subseteq \task.\\
	&\hspace{2.8cm}\world(t) - \world(\max\setCond{t' \in \time}{t' < t \land \taskSat{\task'}{\world}{t'}}) > dl_{\task}
\end{align*}
An overdue stream should be evaluated, regardless of the assigned priority.
We translate a static schedule constraint $\staticSchedule \in \PrioritySchedule$ into a dynamic schedule constraint $\dynamicSchedule \in \Schedule$ with:
\begin{align*}
	\dynamicSchedule(\world, t)(\task) = \begin{cases}
	\scheduleYes & \text{if } \exists \task' . \overdue(\task, \world, t+1) \land \neg \overdue(\task', \world, t+1)\\
	& \phantom{\text{if }} \land \taskSat{\task'}{\world}{t+1}\\
	\scheduleYes & \text{if } (\_, p_1) = \lastPrio_s(\task, \world, t)\\
	& \phantom{if } \land \exists \task'. (\_, p_2) = \lastPrio_s(\task', \world, t)\\
	& \phantom{if } \land \taskSat{\task'}{\world}{t+1} \land p_1 > p_2 \land \neg\overdue(\task', \world, t+1)\\
	& \phantom{if } \land \neg\exists \task'' \succeq \task. \taskSat{\task''}{\world}{t+1}\\
	\scheduleMaybe & \text{otherwise}
	\end{cases}
\end{align*}
This definition assigns overdue tasks an even higher priority, but in general, it follows the previous definition.

\subsection{Valid Scheduler}

A scheduler is a program that, given schedule and bandwidth constraints, decides at each time point which set of tasks to evaluate.
If all selections respect the constraints, the scheduler is considered valid:
\begin{definition}[Valid Scheduler]\label{def:valid_scheduler}
	Given a static schedule constraint $\staticSchedule$ over a set of $\Tasks$, a specification $\varphi$ and a bandwidth bound $B$, a scheduler $\SA_{\staticSchedule,B} : \World \times \time \rightarrow \powerset{\Tasks}$ that decides at each timepoint which streams are evaluated at the next time step, is \emph{valid} if
	\begin{enumerate}
		\item
		$\begin{aligned}[t]
			&\forall \world \in \semantics{(\varphi, \staticSchedule, B)}. \forall t \in \time. \forall \tasks \subseteq \Tasks. \forall i \in \inputValues.\\
			&\quad \prefix_{\SA_{\staticSchedule,B}}(\world, t) \land \SA(\world, t) = \tasks \land i \models \tasks\\
			&\quad\rightarrow \exists \world' \in \semantics{(\varphi, \staticSchedule, B)}. \validTasks(\world', t+1, \tasks) \land \world'[..t] = \world[..t] \land \world'[t+1] = i.
		\end{aligned}$
		\item 
		$\begin{aligned}[t]
			&\forall \world \in \semantics{(\varphi, \staticSchedule, B)}. \forall t \in \time. \exists t' > t. |\SA_{\psi,B}(\world, t')| \ge 1
		\end{aligned}$
	\end{enumerate}
	with
	\begin{align*}
        \validTasks &: \World \times \time \times \powerset{\Tasks} \rightarrow \BB\\
		\validTasks(\world, t, \tasks) &= \forall \task \in tasks. \taskSat{\task}{\world}{t} \land \forall \task \notin \tasks . \taskUnSat{\task}{\world}{t}\\
		&\quad\; \land \forall \task_1,\task_2 \in \Tasks. \task_1 < \task_2 \land \task_1 \in \tasks \rightarrow \task_2 \in \tasks\\
        \prefix_{\SA_{\staticSchedule,B}} &: \World \times \time \rightarrow \BB\\
		\prefix_{\SA_{\staticSchedule,B}}(\world, t) &= \forall t' < t. \validTasks(\world, t'+1, \SA_{\staticSchedule,B}(\world, t'))
	\end{align*}
\end{definition}

A scheduler is considered valid if it satisfies two properties:
\begin{enumerate*}
    \item The scheduler must never get stuck.
    That is, for every point in time and any possible choice of input values consistent with the selected tasks, the evaluation model must be able to continue as a correctly scheduled evaluation model.
    \item The scheduler may not indefinitely select empty tasks.
\end{enumerate*}
These conditions ensure that the scheduler defines a sound evaluation strategy.

%% file: implementation.tex
\section{Active Scheduling in RTLola}
\label{sec:rtlola_scheduling}
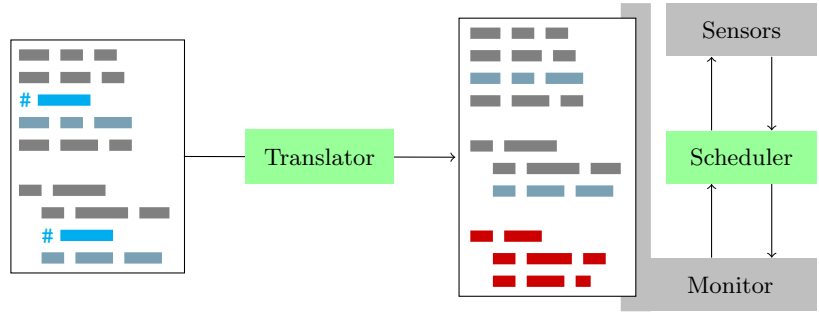
\begin{figure}[t]
\hspace*{5mm}\scalebox{1.0}{
\begin{tikzpicture}[codeline/.style={line width=1.5mm,gray},annotation/.style={codeline,cyan},annotated/.style={codeline,cyan!30!gray},additional/.style={codeline,red!80!black}]
	\def\sep{-0.3}
	\def\restlen{2cm}
	\begin{scope}[local bounding box=in_spec]
		\draw[codeline,dash pattern=on 4mm off 1.5mm on 3mm off 1.5mm on 3mm off 10cm] (0,0) -- (\restlen,0);
		\draw[codeline,dash pattern=on 4mm off 1.5mm on 4mm off 1.5mm on 3mm off 10cm] (0,1*\sep) -- (\restlen,1*\sep);
		\draw[annotation,dash pattern=on 7mm off 10cm,xshift=2.5mm] (0,2*\sep) node[left=-1mm] {\lstinline|#|} -- (\restlen,2*\sep);
		\draw[annotated,dash pattern=on 4mm off 1.5mm on 3mm off 1.5mm on 5mm off 10cm] (0,3*\sep) -- (\restlen,3*\sep);
		\draw[codeline,dash pattern=on 4mm off 1.5mm on 5mm off 1.5mm on 3mm off 10cm] (0,4*\sep) -- (\restlen,4*\sep);
		\draw[codeline,dash pattern=on 3mm off 1.5mm on 7mm off 10cm] (0,6*\sep) -- (\restlen,6*\sep);
		\draw[codeline,dash pattern=on 3mm off 1.5mm on 7mm off 1.5mm on 4mm off 10cm,xshift=3mm] (0,7*\sep) -- (\restlen,7*\sep);
		\draw[annotation,dash pattern=on 7mm off 10cm,xshift=5.5mm] (0,8*\sep) node[left=-1mm] {\lstinline|#|} -- (\restlen,8*\sep);
		\draw[annotated,dash pattern=on 3mm off 1.5mm on 5mm off 1.5mm on 5mm off 10cm,xshift=3mm] (0,9*\sep) -- (\restlen,9*\sep);
	\end{scope}

	\draw ([yshift=2mm]in_spec.north west) rectangle ([yshift=-2mm]in_spec.south west -| 22mm,);

	\begin{scope}[xshift=6cm,yshift=2.9mm,local bounding box=out_spec]
		\draw[codeline,dash pattern=on 4mm off 1.5mm on 3mm off 1.5mm on 3mm off 10cm] (0,0) -- (\restlen,0);
		\draw[codeline,dash pattern=on 4mm off 1.5mm on 4mm off 1.5mm on 3mm off 10cm] (0,1*\sep) -- (\restlen,1*\sep);
		\draw[annotated,dash pattern=on 4mm off 1.5mm on 3mm off 1.5mm on 5mm off 10cm] (0,2*\sep) -- (\restlen,2*\sep);
		\draw[codeline,dash pattern=on 4mm off 1.5mm on 5mm off 1.5mm on 3mm off 10cm] (0,3*\sep) -- (\restlen,3*\sep);
		\draw[codeline,dash pattern=on 3mm off 1.5mm on 7mm off 10cm] (0,5*\sep) -- (\restlen,5*\sep);
		\draw[codeline,dash pattern=on 3mm off 1.5mm on 7mm off 1.5mm on 4mm off 10cm,xshift=3mm] (0,6*\sep) -- (\restlen,6*\sep);
		\draw[annotated,dash pattern=on 3mm off 1.5mm on 5mm off 1.5mm on 5mm off 10cm,xshift=3mm] (0,7*\sep) -- (\restlen,7*\sep);
		\draw[additional,dash pattern=on 3mm off 1.5mm on 5mm off 10cm] (0,9*\sep) -- (\restlen,9*\sep);
		\draw[additional,dash pattern=on 3mm off 1.5mm on 6mm off 1.5mm on 3mm off 10cm,xshift=3mm] (0,10*\sep) -- (\restlen,10*\sep);
		\draw[additional,dash pattern=on 3mm off 1.5mm on 5mm off 1.5mm on 2mm off 10cm,xshift=3mm] (0,11*\sep) -- (\restlen,11*\sep);
	\end{scope}

	\draw ([yshift=2mm,xshift=-1.5mm]out_spec.north west) coordinate (top_left) rectangle ([yshift=-2mm]out_spec.south west -| 82mm,) coordinate (bottom_right);

	\draw[->] ([xshift=2.3cm]in_spec.west) -- node[fill=green!40!white,inner sep=2.5mm] {Translator} ([xshift=-2mm]out_spec.west);

	\begin{scope}[on background layer]
	\fill[lightgray] (bottom_right) -| ($(bottom_right) - (2mm,2mm)$) -- ++(4mm,0) coordinate (gray_bottom_right) -- ([xshift=2mm,yshift=2mm]top_left -| bottom_right) coordinate(gray_top_right) -| ([xshift=-2mm]top_left -| bottom_right) -- (top_left -| bottom_right) -- cycle;
	\end{scope}

	\begin{scope}[transform canvas={xshift=-1mm}]
	\node[anchor=south west,fill=lightgray,minimum width=2.3cm,minimum height=7mm] (monitor) at (gray_bottom_right) {Monitor};
	\node[anchor=north west,fill=lightgray,minimum width=2cm,minimum height=7mm] (sensors) at ([xshift=3mm]gray_top_right) {Sensors};
	\node[fill=green!40!white,minimum width=2cm,minimum height=7mm] (scheduler) at ($(monitor -| sensors)!0.5!(sensors)$) {Scheduler};
	\draw[->,transform canvas={xshift=4mm}] (scheduler) -- (scheduler |- monitor.north);
	\draw[->,transform canvas={xshift=-4mm}] (scheduler |- monitor.north) -- (scheduler);
	\draw[->,transform canvas={xshift=4mm}] (scheduler |- sensors.south) -- (scheduler);
	\draw[->,transform canvas={xshift=-4mm}] (scheduler) -- (scheduler |- sensors.south);
	\end{scope}
\end{tikzpicture}
}
\caption{
	Overview of the scheduling process.
}
\label{fig:impl_overview}
\end{figure}

This section demonstrates how our scheduling approach is integrated with the stream-based monitoring language \rtlola.
The overall architecture is illustrated in \Cref{fig:impl_overview}.
First, we allow users to augment \rtlola specifications with scheduling annotations, as shown on the left side of the figure.
These annotations assign deadlines or priorities to individual streams or clauses, and are detailed in \Cref{sec:annotations}.
A translator then processes the annotated specification and produces a transformed \rtlola specification.
During this step, the translator interprets the scheduling annotations and adds helper streams that encode the scheduling constraints.
We describe the translation process in \Cref{sec:translation} in more detail.
The resulting specification is compatible with existing \rtlola implementations, while a separate scheduling component interfaces between the backend and the sensors to issue sensor queries at runtime.
We further explain this scheduler interface in \Cref{sec:scheduler}.

\subsection{Annotations}
\label{sec:annotations}
\noindent For the configuration of the scheduler, we embed the scheduling-related information as annotations directly within the specification.
These annotations attach constraints to streams -- either directly on input streams or on the \lstinline!eval!-clause of output streams.
Each annotation defines either a \emph{priority} or a \emph{deadline} -- following the idea from \Cref{sec:static_schedule}.
We represent the annotations as a mapping $a \in \Sref \rightarrow \powerset{\condition \times \VV}$.
The conditions are a combination of the pacing of the stream and an expression derived from the \lstinline!when!-conditions of the annotated eval clause.
However, since \lstinline!when!-conditions are evaluated top-to-bottom, the effective condition for each clause includes not only its own \lstinline!when!-condition but also the negation of all preceding ones.
This ensures that conditions are mutually exclusive, reflecting the semantics of sequential \lstinline!when! evaluation.
Depending on the kind of static schedule used for scheduling, the value $\VV$ is either a deadline $\in \RR$ or a priority $\in \NN$.
For deadline-priority scheduling, annotations consist of priorities, while deadlines of individual tasks are configured globally.

\begin{example}
The specification in \Cref{fig:example_annotations} monitors two conditions:
\begin{enumerate*}
	\item whether the current latitude and longitude violate a geofence, and
	\item whether the current altitude exceeds an upper bound.
\end{enumerate*}
The annotations result in the following representation:
{\footnotesize
\begin{alignat*}{2}
	s(\text{\lstinline!bound_violation!}) = &\; \{&&(\text{\lstinline!|@lat&&lon|!}, \text{\lstinline!distance_to_bound < 12.0!}), 10),\\
	& &&(\text{\lstinline!|@lat&&lon|!}, \text{\lstinline!distance_to_bound >= 12.0!}), 1)\},\\
	s(\text{\lstinline!altitude_violation!}) = &\; \{&&(\text{\lstinline!|@alt|!}, \text{\lstinline!true!}), 5)\}
\end{alignat*}}
The specification assigns the output \lstinline!bound_violation! a high priority, numerically represented as 10, if the distance to the bound is smaller than 12.
If it is further away from the bound, the bound check is assigned a low priority, a 1.
The altitude check is assigned a constant medium priority, represented as a 5.
As a result, when the system is far away from a geofence boundary, the altitude has a higher priority and is therefore checked more frequently.
Conversely, as the system approaches a potential geofence violation, geofence monitoring is prioritized.
\end{example}

\begin{figure}[t]
\begin{lstlisting}
input lat, lon, alt : Float64
output distance_to_bound
    eval |@lat&&lon|
          with min(lat - 3, UPPER_LAT - 8, lon - 5, UPPER_LON - 10)
output bound_violation
	#[priority="high"]
	eval |@lat&&lon| when distance_to_bound < 12.0
          with distance_to_bound < 0.0
	#[priority="low"]
	eval |@lat&&lon| when distance_to_bound >= 12.0
          with false
output altitude_violation
	#[priority="medium"]
	eval |@alt| with altitude > 50.0
\end{lstlisting}

\caption{Example of scheduling annotations in \rtlola.}
\label{fig:example_annotations}
\end{figure}

\subsection{Translation}
\label{sec:translation}

A translator converts the annotated specification into a standard \rtlola specification that any existing implementation can process.
This translation aims to ensure that the resulting monitor satisfies the specified scheduling constraints by guiding the scheduling component to supply the necessary inputs at appropriate times.
For this, it adds additional output streams to the resulting specification, which are in turn read by the scheduling component.

\setlength\intextsep{0pt}
\begin{wrapfigure}[10]{r}{0.5\linewidth}
\begin{lstlisting}
	#[priority="high"]
	input a : UInt64
	#[priority="medium"]
	input b : UInt64
	output c
		#[priority="low"]
		eval |@a| when true with a + 1
\end{lstlisting}
\caption{Example specification.}
\label{fig:problem_spec_prio}
\end{wrapfigure}
At first glance, one expects the scheduling annotations to correspond directly to the static schedule.
However, setting $\Tasks = \Sref$ leads to an issue, as illustrated in \Cref{fig:problem_spec_prio}.
Assume the scheduler is restricted to providing a new value to only one input stream at a time.
In this scenario, there is no valid schedule that updates any inputs.
If stream $a$ receives a new value due to its high priority, stream $c$ is also evaluated, triggered by its pacing.
Yet, a stream with higher priority, namely $b$, is not evaluated, contradicting the scheduling semantics.

The root of the issue is, that while output streams can be annotated with scheduling constraints, in practice, the evaluation of output streams is controlled by the underlying monitor and determined by its pacing type -- the scheduler can only decide when to emit new input values to the monitor -- and not prevent the evaluation of outputs.
Consequently, the scheduler must supply inputs at a frequency that respects the scheduling annotations of all outputs streams, by satisfying their pacing at the right times.
To support this idea, we restrict tasks to sets of input streams, i.e., $\Tasks = \powerset{\inputref}$, and propagate all annotations from output streams to the corresponding sets of inputs necessary for their evaluation.
\begin{definition}[\rtlola Tasks]
	Given an \rtlola specification $\varphi$ containing input streams $\inputref$, we define the task set $\Tasks = \powerset{\inputref}$ with
	\begin{align*}
		\taskSat{\task}{\world}{t} \quad&\iff\quad \forall i \in \task. \world(i)(t) \ne \bot\\
        \task_1 \preceq \task_2 \quad&\iff\quad \task_1 \subseteq \task_2\\
        i \models \tasks \quad&\iff\quad \forall s \in \inputref . i(s) \ne \bot \Leftrightarrow \set{i} \in \tasks.
	\end{align*}
\end{definition}

Given annotations $a : \Sref \rightarrow \powerset{\condition \times \VV}$, we construct the static schedule $\staticSchedule_a$ as the following, where $\mathit{pac}(s)$ returns the pacing of stream $s \in \Sref$:
\begin{align*}
	\staticSchedule_a &: \powerset{\inputref} \rightarrow \powerset{\condition \times \VV}\\
	\staticSchedule_a(\task) &= \big\{(c, v)\;\mid\;v = \max\big\{v'\;\mid\;\exists s \in \Sref. \mathit{pac}(s) \subseteq \task \land (c', v') \in a(s) \land c' \Rightarrow c\big\}\big\}
\end{align*}
In other words, for each task -- i.e., set of input streams -- the schedule includes all constraint-condition pairs $(c,v)$ such that $v$ is the maximal value out of all annotations of streams whose pacing is implied by $\task$, and whose conditions imply $c$.
The maximal value is chosen according to the domain $\VV$: it represents the most restrictive constraint, e.g., the highest priority, or the shortest deadline.

Through pacing types in \rtlola, the task dependency relation is not sufficient to determine if a task set is valid.
Consider the following example where we have two input streams, $i$ and $i'$, and a task set containing the tasks $i$ and $i'$.
When both input streams are updated, the pacing consisting of their combination also activates because of the \rtlola semantics.
Therefore, we need to ensure that the task containing the combination is also part of the task set.
We therefore have to strengthen the notion of valid task sets introduced in \Cref{def:valid_scheduler} with the additional constraint
\[
    \forall \task_1,\task_2 \in \Tasks. \task_1,\task_2 \in \tasks \rightarrow \task_1 \cup \task_2 \in \tasks.
\]

To represent this static schedule as additional \rtlola streams, we construct a single stream for each task, i.e., each pacing type.
We utilize different clauses guarded by the corresponding condition to dynamically assign the schedule value of a task based on the current state.
The clauses are ordered by their schedule value in ascending order, so that more restrictive schedules are considered first.
Further, for each task, separate output streams are added, which note the time of the task's last evaluation, and if needed, additional streams to indicate if a task is overdue.
For the specification in \Cref{fig:example_annotations}, the translation process would add the following output streams to the specification:
\begin{lstlisting}
	output schedule_lat_lon
		eval |@lat&&lon| when distance_to_bound < 10.0 with 10
		eval |@lat&&lon| when distance_to_bound >= 10.0 with 1
    output last_lat_lon eval |@lat&&lon| with now
	output schedule_alt eval |@alt| with 5
    output last_alt eval |@alt| with now
\end{lstlisting}

\subsection{Scheduler}
\label{sec:scheduler}

While all previous steps occur before the monitor is executed, the scheduler is the runtime component responsible for querying new inputs and passing them to the underlying monitor.
In each cycle, the scheduler queries sensors and forwards the obtained values to the monitor, triggering a new evaluation step.
As usual, the monitor computes new output stream values, including the schedule streams that guide the scheduling decisions.
The scheduler uses these newly computed values to decide which inputs to query in the next cycle.

The strategy the scheduler uses to populate the available bandwidth depends on the type of scheduling constraints.
For all strategies, the bandwidth is specified as the number of input values per event $b$ and the number of events per second $f_e$.
At runtime, the scheduler emits events to the monitor at the specified frequency and populates each event with the predefined number of inputs according to the scheduling strategy.

For deadline scheduling, the scheduler $DS_{\staticSchedule,B}$ employs an earliest-deadline-first strategy~\cite{horn1974some}.
In this mode, each task is associated with a deadline representing the maximum allowable timestamp of the next update.
The scheduler tracks the current time and selects the task with the most urgent deadlines first until the bandwidth is fully utilized.
In contrast, for priority scheduling, the scheduler $PS_{\staticSchedule,B}$ fills the event with tasks, selected in decreasing order of priority.
It is the same for the deadline-priority scheduler $DPS_{\staticSchedule,B}$, only that overdue tasks are assigned to a new, highest priority.
If multiple streams share the same priority level, the scheduler resolves this by choosing the stream that was updated the longest time ago.
\ifthenelse{\boolean{fullversion}}{
The formal definition of these schedulers can be found in~\Cref{app:schedulers}.
}{
The formal definition of these schedulers can be found in the full version.
}
With some restrictions, these schedulers satisfy the validity conditions in \Cref{def:valid_scheduler}:
\begin{theorem}[Valid \rtlola Schedulers]
\label{theorem:valid-scheduler}
    Given a well-defined specification $\varphi$, an annotation $a \in \Sref \rightarrow \powerset{\condition \times \VV}$, and a bound $B = \simpleBound{b}$.
    \begin{itemize}
        \item The schedulers $PS_{S_a,B}$ and $DPS_{S_a, B}$ are valid if $b \ge \max_{\task \in \Tasks} |\task|$,
        \item The scheduler $DS_{\staticSchedule_a,B}$ is valid if $\forall \task \in \Tasks. \forall (c,dl) \in \staticSchedule(\task). dl > n$ and $b \ge \max_{\task \in \Tasks} |\task|$, with:
        \begin{align*}
		n &= \min \setCond{n'}{\forall \tasks \in \mathit{Permutations}_{|\Tasks|}(\Tasks).splits_B(\tasks, \varepsilon) = n'}
        \end{align*}
        \begin{align*}
        \mathit{Permutations}_{|\Tasks|}(\Tasks) &= \setCond{\task_1\task_2\ldots \in \Tasks^{|\Tasks|}}{\forall i,j. i = j \lor \task_i \ne \task_j}\\
        splits_B(\varepsilon, cur) &= 1\\
		splits_B(\task_1 \circ \tasks, cur) &= \begin{cases}
			1 + splits_B(\task_1 \circ \tasks, \varepsilon) & \text{if } |cur \cup \task_1| > b\\
			splits_B(\tasks, cur \cup \task_1) & \text{otherwise}
		\end{cases}
	\end{align*}
    \end{itemize}
\end{theorem}
Generally, the bandwidth bound $b$ needs to be large enough to accommodate each task individually.
For the priority and deadline-priority schedulers, in each step, the tasks can be selected purely on the current priority assignments, and there are no constraints from previous evaluations.
The schedule selects the highest priority tasks.
Our construction can trigger tasks with a potentially lower priority; however, the dependency relation, together with the construction through the scheduling annotations, ensures that these priorities do not conflict with the semantics.
For the deadline scheduler, we need to account for constraints from previous cycles.
There, the bound ensures that previously imposed deadlines can't conflict with newly added ones.
\ifthenelse{\boolean{fullversion}}{
Detailed proof sketches for the validity of each scheduler can be found in~\Cref{app:schedulers}.
}{
Detailed proof sketches for the validity of each scheduler can be found in the full version.
}


%% file: evaluation.tex
\section{Implementation \& Evaluation}
\label{sec:evaluation}


We have implemented our approach from \Cref{sec:rtlola_scheduling} on top of the \rtlola monitoring framework~\cite{baumeister2024tutorial}.
For the evaluation, we use the AirSim simulator, a simulation environment for drones and cars developed by Microsoft.
AirSim exposes various sensors through an API, including GPS coordinates, altitude readings, camera images, and LIDAR data.

For our evaluation, we investigate the following setting with the deadline-priority constraints.
A drone is tasked to collect barometer data during flight for an experiment in a restricted airspace.
It relies on four sensors, which must share the limited available bandwidth of the drone: the GPS and altitude sensor for enforcing their respective bounds, and two barometers, whose data collection for the experiment is the primary target of the flight.
We use stream annotations as presented in \Cref{sec:annotations} to represent different priorities for the tasks in the specification.
The closer the drone comes to the geofence border, threatening a violation, the higher the assigned priority of the geofence task becomes.
Likewise, the closer the drone comes to an altitude violation, the higher the assigned priority.
The barometers, which serve the experiment, are assigned a constant medium priority.
As a result, they receive most of the bandwidth if there is no immediate danger of a boundary violation.
We employ overdue deadlines to ensure the geofence bounds are checked at least every 3 seconds to prevent violations from being missed because of higher priority tasks.
\ifthenelse{\boolean{fullversion}}{
The annotated specification can be found in \Cref{app:eval_spec}.
}{
The annotated specification can be found in the full version.
}


In our experiments, we compare the behavior of our scheduled monitoring approach against a set of baseline monitors that query all inputs at fixed frequencies.
Each baseline monitor queries all sensors at the same frequency, but we vary this frequency across baselines to explore the trade-off between bandwidth and responsiveness.
To ensure a fair comparison, all monitors observe the same drone flight in parallel.

The goal of the evaluation is to assess the effectiveness of actively scheduling the inputs in runtime monitoring under bandwidth constraints.
Specifically, we aim to answer the following questions:
\begin{enumerate*}
\item Can we reduce the overall bandwidth consumption while maintaining a comparable level of monitoring quality?
\item Can we detect specification violations earlier by allocating the bandwidth more intelligently -- focusing on data that is more likely to reveal a violation?
\end{enumerate*}

\begin{figure}[t]
	\begin{subfigure}[t]{0.498\linewidth}
	\includegraphics[width=\linewidth]{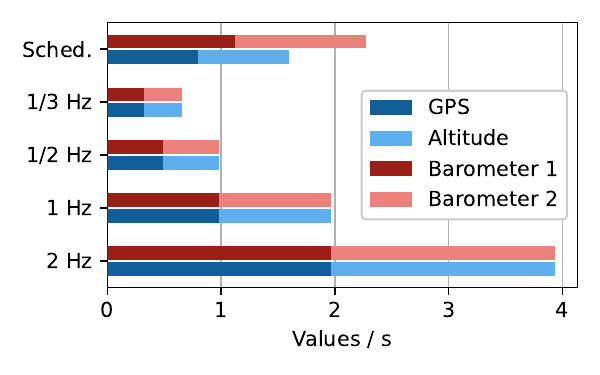}
	\caption{Bandwidth consumption of each sensor.}
	\label{fig:bandwidth_comparison}
	\end{subfigure}
	\hfill
	\begin{subfigure}[t]{0.498\linewidth}
	\includegraphics[width=\linewidth]{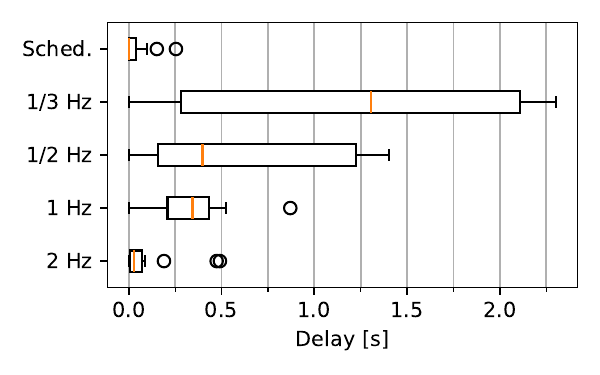}
	\caption{Delay until trigger is detected.}
	\end{subfigure}
	\caption{Comparison between the scheduled monitor and four fixed-frequency ones.}
	\label{fig:trigger_delay_comparison}
	\vspace*{-5mm}
\end{figure}

We evaluate the performance of our scheduling approach over 10 flights with the simulator by analyzing two metrics obtained from the evaluation.
\Cref{fig:bandwidth_comparison} shows the average number of values that occupy the bandwidth for each sensor.
For each monitor, we group the sensors into two categories, represented by the horizontal bars:
The blue bar captures the bandwidth used for geofence and altitude checks, while the red bar shows the bandwidth used for collecting barometer data for the experiment.
In the fixed-frequency monitors, bandwidth is distributed evenly across all sensors.
In contrast, the scheduled monitor dynamically adjusts the frequency of sensor queries based on the current state of the system, resulting in an uneven allocation of bandwidth.
As shown in the figure, the scheduled monitor allocates more bandwidth to the barometer sensors, aligning with the experiment's objective of maximizing barometer data collection.
Meanwhile, the sensors responsible for enforcing the geofence and altitude boundaries are queried less frequently.
In contrast, \Cref{fig:trigger_delay_comparison} presents a boxplot of the delay in detecting violations.
Delays are measured relative to the earliest point in time any monitor detects the same violation.

To answer the first question, the results show that our scheduled monitor detects violations as early as the fastest fixed-frequency monitor operating at 2~Hz.
Notably, it achieves this while consuming half the bandwidth, as evident from the comparison in the previous figure.

To answer the second question, we compare the scheduled monitor against a fixed-frequency monitor with equal bandwidth usage.
The scheduled monitor consumes four values per second -- equivalent to the bandwidth of a 1~Hz fixed-frequency monitor.
Despite this, the scheduled monitor detects violations significantly earlier than the 1~Hz monitor, because in critical situations, the scheduled monitor focuses on the more critical inputs.
This demonstrates that our approach makes more effective use of the available bandwidth through our scheduling approach.

%% file: conclusion.tex
\section{Conclusion}

We addressed the challenge of runtime monitoring in bandwidth-constrained environments by presenting an approach to dynamically allocating the available bandwidth to different sensors depending on the monitor's current state.
This approach enables the monitor to prioritize inputs and allocate bandwidth to data most likely to reveal specification violations.
To facilitate this, we define a formal semantics for scheduled monitors and introduce several static schedules.
Scheduling annotations, embedded directly within the specification, offer an intuitive and flexible mechanism to express constraints guiding the scheduler.
A scheduler interfacing between the monitor and the sensors is responsible for querying sensors at the appropriate times and passing the values to the underlying monitor.
We evaluated our approach using simulations in the AirSim drone simulator. 
The results demonstrate the effectiveness of scheduled monitoring in detecting violations early on, while consuming significantly less bandwidth compared to a monitor receiving all inputs at a fixed frequency.

For future work, we want to evaluate our approach in a real-world setting to assess the overhead introduced by querying sensors with the active role of the monitor.
Furthermore, many real-world settings involve sensors producing multiple values, leading to more complex task dependencies.
We want to investigate how our approach can be extended to handle such cases.


%% file: appendix.tex
\section{Semantics}

\begin{definition}[Evaluation Model]
\begin{align*}
	\stream &: \time \rightarrow \VV_\bot\\
	\streammap &: \Sref \rightarrow \stream\\
	\timemap &: \time \rightarrow \RR\\
	\World &: \streammap \times \timemap
\end{align*}
The evaluation model includes a mapping $\streammap$ that assigns each input and output stream reference to a stream -- a sequence of values indexed by discrete time steps.
At any time $t\in \time$, a stream may either contain a value or $\bot$, indicating that the stream does not contain a value at that time (e.g., due to pacing or a filter condition not being satisfied).
The $\timemap$ assigns each discrete time step a real-valued timestamp.
\end{definition}
For convenience, we use the following shorthand notation for $\world = (\mathit{streams}, \mathit{times})$:
\begin{itemize}
	\item $\world(t) := \mathit{times}(t)$ gives the real-time timestamp at step $t \in \time$, and
	\item $\world(\sid)(t) := \mathit{streams}(\sid)(t)$ gives the value of stream $\sid \in \Sref$ at time $t\in\time$.
\end{itemize}

\section{Specification}
\label{app:eval_spec}

\begin{lstlisting}
#![frequency="2Hz",bound="2"]
import math
#[deadline="3s"]
input gps_lat_long : (Float64, Float64)
#[deadline="3s"]
input gps_altitude : Float64

#[priority="medium",deadline="3s"]
input barometer_pressure : Float64
#[priority="medium",deadline="3s"]
input barometer_altitude : Float64

output lat := gps_lat_long.0
output long := gps_lat_long.1
output start_lat |@gps_lat_long| := start_lat.offset(by:-1).defaults(to: lat)
output start_long |@gps_lat_long| := start_long.offset(by:-1).defaults(to: long)
output start_altitude |@gps_altitude| := start_altitude.offset(by:-1).defaults(to: gps_altitude)

output distance_to_start := sqrt((lat-start_lat)*(lat-start_lat)
	+ (long-start_long)*(long-start_long))*10000.0
output altitude_above_ground := gps_altitude - start_altitude

output geofence := distance_to_start $\ge$ 8.0

output scheduled_geofence
	#[priority="low"]
    eval |@gps_lat_long| when distance_to_start $\le$ 4.0
						  with geofence
	#[priority="medium"]
    eval |@gps_lat_long| when distance_to_start $\le$ 6.0
						  with geofence
	#[priority="high"]
    eval |@gps_lat_long| with geofence

trigger scheduled_geofence "outside geofence"

output altitude_bound := altitude_above_ground $\ge$ 10.0

output scheduled_altitude_bound
	#[priority="low"]
    eval |@gps_altitude| when altitude_above_ground $\le$ 5.0
						  with altitude_bound
	#[priority="medium"]
    eval |@gps_altitude| when altitude_above_ground $\le$ 7.5
						  with altitude_bound
	#[priority="high"]
    eval |@gps_altitude| with altitude_bound

trigger scheduled_altitude_bound "altitude too high"
\end{lstlisting}

\section{\rtlola Schedulers}
\label{app:schedulers}
A Scheduler is formally defines as:
\begin{definition}[Scheduler]
	Given $\task_i,\task_j \in \Tasks$, we define total orders
	\begin{align*}
		\task_i \le_{DL} \task_j &\Leftrightarrow le(\task_i) + v(\task_i) \le le(\task_j) + v(\task_j)\\
		\task_i \le_{P} \task_j &\Leftrightarrow \begin{cases}
			v(\task_i) \le v(\task_j) & \text{if } v(\task_i) \ne v(\task_j)\\
			le(\task_i) \le le(\task_j) & \text{otherwise}
		\end{cases}\\
		\task_i \le_{DP} \task_j &\Leftrightarrow \begin{cases}
			od(\task_i) \land \neg od(\task_j) & \text{if } od(\task_i) \ne od(\task_j)\\
			v(\task_i) \le v(\task_j) & \begin{aligned}\text{if } &od(\task_i) = od(\task_j) \\&\land v(\task_i) \ne v(\task_j)\end{aligned}\\
			le(\task_i) \le le(\task_j) & \text{otherwise}
		\end{cases}
	\end{align*}
	where $v(\task)$, $le(\task)$, and $od(\task)$ are shorthand notations for stream-accesses to the helper streams holding the schedule value, the time of last evaluation and their overdue status respetively.

	Given a schedule $\staticSchedule$, a bound $B$, and a total order $\le$, we define
	\[
		PS_{\staticSchedule,B,\le} = \takeNextEvent_B(\mathit{sort}_{\le}(\Tasks))
	\]
	with 
	\begin{align*}
		&\takeNextEvent_B(\tasks, cur) =\begin{cases}
			\takeNextEvent_B(tail, cur \cup task_1) & \text{if } \tasks \ \task_1 \circ tail \land |cur \cup \task_1| < B\\
			\setCond{\task \in \Tasks}{\task \subseteq cur} &\text{otherwise}
		\end{cases}
	\end{align*}
	The schedulers are then defined as $DS_{\staticSchedule,B}=PS_{\staticSchedule,B,\le_{DL}}$, $PS_{\staticSchedule,B}=PS_{\staticSchedule,B,\le_P}$, and $DPS_{\staticSchedule,B}=PS_{\staticSchedule,B,\le_{DP}}$.
\end{definition}

Then, we can first proof that these schedulers only produce valid task sets:
\begin{theorem}[Valid \rtlola Task Sets]
\label{theorem:valid-task-set}
	Given a well-defined specification $\varphi$, an annotation mapping $a$, and a bound $b \ge \max_{\task \in \Tasks} |\task|$ with $B = \simpleBound{b}$, the schedulers $DS_{\staticSchedule_a,B}$, $PS_{\staticSchedule_a,B}$, and $DPS_{\staticSchedule_a,B}$ produce at every time point valid task sets.
\end{theorem}
\begin{proof}
	The proof follows directly from the construction of the task sets from the takeUntil definition and the relations $\le_{DL}$, $\le_{P}$, and $\le_{DP}$.
\end{proof}

To proof \Cref{theorem:valid-scheduler}, we proof that each scheduler individually is valid, but first we introduce two side lemmas:
\begin{lemma}\label{lemma:1}
Given a well-defined specification $\varphi$, an evaluation model $\world \in \semantics{\varphi}$, an annotation mapping $a$, and a bound $B = \simpleBound{b}{}$ with $b \ge \max_{\task \in \Tasks} |\task|$, and $\forall \task \in \Tasks. \forall (c,dl) \in \staticSchedule_a(\task). dl > n$ where
\begin{align*}
    n &= \max \setCond{n'}{\forall \task_1 \task_2 \ldots \in \mathit{Permutations}_{|\Tasks|}(\Tasks).splits_B(\task_1 \task_2 \ldots, \varepsilon) = n'}
\end{align*}.
Then for every time $t \in \time$:
\begin{align*}
    &(\forall t' < t. \validTasks(\world, t', DS_{\staticSchedule_a,B}(\world, t')))\\&\rightarrow \mathit{split'}(\mathit{sort}_{\leq_{DL}}(\Tasks)) = e_0\ldots e_m \land \forall \task \in e_j. \mathit{relativeDl}_{\staticSchedule_a}(\task, \world, t) \geq j
\end{align*}, with $\mathit{relativeDl}$ returns the relative timepoint when the least deadline of this task is due and $\mathit{split'}$:
\begin{align*}
        splits'_B(\varepsilon, cur) &= \begin{cases}
            \setCond{\task \in \Tasks}{\task \subseteq cur} \circ \epsilon & \text{if } cur \neq \emptyset\\
            \epsilon & \text{otherwise}
        \end{cases}\\
		splits'_B(\task_0 \circ \tasks, cur) &= \begin{cases}
			\setCond{\task \in \mathit{tail}}{\task \subseteq cur} \circ splits'_B(\task_0 \circ \mathit{tail}, \varepsilon) & \text{if } |cur \cup \task_0| > b\\
			splits'_B(\mathit{tail}, cur \cup \task_0) & \text{otherwise}
		\end{cases}
	\end{align*}
\end{lemma}
\begin{proof}
    We proof this lemma by induction over t:
    \begin{itemize}
        \item[t=0] Since $\forall \task \in \Tasks. \forall (c,dl) \in \staticSchedule_a(\task). dl > n$, we know from the assumption with the minimal deadlines
        \[
        \forall \task \in \Tasks. \mathit{relativeDl}_{\staticSchedule_a}(\task, \world, t) > n
        \]
        since $\task$ is either a new task, where the assumption holds or does not have yet a relative deadline.
    In addition, by construction of $\mathit{split'}$, $\mathit{split'}(\mathit{sort}_{\leq_{DL}}(\Tasks))$ returns a sequence with a size of at most $n$.
    So $\task \in e_j. \mathit{relativeDl}_{\staticSchedule_a}(\task, \world, t) > n \geq j$.
        \item[$t \rightarrow t + 1$] From the induction, we know at time $t$
        \[
        \mathit{split'}(\mathit{sort}_{\leq_{DL}}(\Tasks)) = e_0\ldots e_m \land \forall \task \in e_j. \mathit{relativeDl}_{S_a}(\task, \world, t) \geq j
        \]
        Since $\world$ follows the earliest deadline first algorithm, the inputs at time $t + 1$ follow $e_0$, so $\forall\task \in e_0. \taskSat{\task}{\world}{t + 1}$. If we assume that at $t + 1$ no new deadlines are added, then at time $t + 1$:
        \[
        \mathit{split'}(\mathit{sort}_{\leq_{DL}}(\Tasks)) = e_1\ldots e_me_0
        \] due to the ordering $\le_{DL}$ and by construction it follows directly:
        \[
            \forall \task \in e_j. \mathit{relativeDl}_{\staticSchedule_a}(\task, \world, t) \geq j
        \], since the relative deadline is reduced by one step but so also the position.
        If new deadlines are added at time $t + 1$, then 
        \[
            \forall \task \in e_j. \mathit{relativeDl}_{\staticSchedule_a}(\task, \world, t) \geq \min(n,j) \geq j
        \], since either the task $\task$ is assigned with a new deadline then $\mathit{relativeDl}_{\staticSchedule_a}(\task, \world, t) > n$ or the task keeps its old deadline. Then, this relative deadline is
        \begin{itemize}
            \item either smaller than $n$, then because of the ordering at time $t$, $\task \in e_i$ and $\mathit{relativeDl}_{\staticSchedule_a}(\task, \world, t) \geq i$ and at time $t + a$ it holds $\task \in e'_i = e_{i + 1}$, so $\mathit{relativeDl}_{\staticSchedule_a}(\task, \world, t) \geq i$.
            \item greater or equal $n$, then $\mathit{relativeDl}_{\staticSchedule_a}(\task, \world, t) > n \geq i$
        \end{itemize}  
    \end{itemize}
\end{proof}

\begin{lemma}
\label{lemma:2}
Given a well-defined specification $\varphi$, a world $\world \in \semantics{\varphi}$, an annotation mapping $a$, and a bound $B = \simpleBound{b}{}$ with $b \ge \max_{\task \in \Tasks} |\task|$, and $\forall \task \in \Tasks. \forall (c,dl) \in \staticSchedule_a(\task). dl > n$ where
\begin{align*}
    n &= \max \setCond{n'}{\forall \task_1 \task_2 \ldots \in \mathit{Permutations}_{|\Tasks|}(\Tasks).splits_B(\task_1 \task_2 \ldots, \varepsilon) = n'}
\end{align*}.
Then $\world \in \semantics{\dynamicSchedule_{\staticSchedule_a}}$ if $\forall t. \validTasks(\world, t, DS_{\staticSchedule_a,B}(\world, t')))$.
\end{lemma}
\begin{proof}
    For this proof, we unroll the definition of $\world \in \semantics{\dynamicSchedule_{\staticSchedule_a}}$, so we have to proof:
    \[
	\forall t \in \time. \forall \task \in \Tasks. \begin{cases}
		\taskSat{\task}{\world}{t + 1} &\text{if } \dynamicSchedule_{\staticSchedule_a}(\world, t)(\task) = \scheduleYes\\
		\top &\text{if } \dynamicSchedule_{\staticSchedule_a}(\world, t)(\task) = \scheduleMaybe\\
		\taskUnSat{\task}{\world}{t + 1} &\text{if } \dynamicSchedule_{\staticSchedule_a}(\world, t)(\task) = \scheduleNo
	\end{cases}.
\] with 

\begin{align*}
	\dynamicSchedule_{\Schedule_a}(\world, t)(\task) = \begin{cases}
		\scheduleYes & \text{if } \exists (c, dl) \in \staticSchedule_a(\task) \land \exists t' < t. c \evalTo{\world}{t'} \top\\
		& \phantom{\text{if }} \land \forall t'' \in (t', t]. \taskUnSat{\task}{\world}{t''} \land \world(t+2) > \world(t') + dl\\
		\scheduleMaybe & \text{otherwise}
	\end{cases}
\end{align*}
    Given an arbitrary but fix $t$ and $\task$, then 
    \[
    \dynamicSchedule_{\staticSchedule_a}(\world, t)(\task) = \scheduleYes \rightarrow \taskSat{\task}{\world}{t + 1} 
    \]
    $\dynamicSchedule_{\staticSchedule_a}(\world, t)(\task) = \scheduleYes$ is by construction equivalent to $\mathit{relativeDl}_{\staticSchedule_a}(\task, \world, t) = 1$ and from \Cref{lemma:1}, we know $\forall \task \in e_j. \mathit{relativeDl}_{\staticSchedule_a}(\task, \world, t) \geq j$, so the fixed $\task \in e_0$ otherwise the relative deadline would not be 1 and by the construction of the scheduler $\forall \task \in e_0. \taskSat{\task}{\world}{t + 1}$ 
\end{proof}

\begin{theorem}[Valid Deadline Scheduler]
    \label{theorem:valid-deadline-scheduler}
    Given a well-defined specification $\varphi$, an annotation mapping $a$, and a bound $B = \simpleBound{b}{}$.
    The scheduler $DS_{\staticSchedule_a,B}$ is valid if $\forall \task \in \Tasks. \forall (c,dl) \in \staticSchedule_a(\task). dl > n$ and $b \ge \max_{\task \in \Tasks} |\task|$ with 
    \begin{align*}
    n &= \min \setCond{n'}{\forall \task_1 \task 2 \ldots \in \mathit{Permutations}_{|\Tasks|}(\Tasks).splits_B(\task_1 \task_2 \ldots, \varepsilon) = n'}
\end{align*}.
\end{theorem}
\begin{proof}
    By definition, we have to proof the following conditions:
    \begin{enumerate}
		\item
		$\quad\begin{aligned}[t]
			&\forall \world \in \semantics{(\varphi, \dynamicSchedule_{\staticSchedule_a}, B)}. \forall t \in \time. \forall \tasks \subseteq \Tasks. \forall i \in \inputValues.\\
			&\quad \prefix_{DS_{\staticSchedule_a,B}}(\world, t) \land DS_{\staticSchedule_a, B}(\world, t) = \tasks \land i \models \tasks\\
			&\quad\rightarrow \exists \world' \in \semantics{(\varphi, S_a, B)}. \validTasks(\world', t+1, \tasks) \land \world'[..t] = \world[..t] \land \world'[t+1] = i.
		\end{aligned}$
		\item 
		$\quad\begin{aligned}[t]
			&\forall \world \in \semantics{(\varphi, \dynamicSchedule_{\staticSchedule_a}, B)}. \forall t \in \time. \exists t' > t. |DS_{\staticSchedule_a,B}(\world, t')| \ge 1
		\end{aligned}$
	\end{enumerate}
    We proof each condition separately:
    \begin{enumerate}
        \item We proof this condition by constructing $\world'$, with $\world'[..t] = \world[..t]$, $\world'[t+1] = i$, and $\world'$ follows the scheduler $DS_{S_a, B}$. Then we have to proof:
        \begin{itemize}
            \item $\world \in \semantics{\varphi}$ follows from the well-definedness of $\varphi$
            \item $\world \in \semantics{\dynamicSchedule_{\staticSchedule_a}}$ follows from \Cref{lemma:2}
            \item $\forall t \in \time. B(\world, t)$ and $\validTasks(\world', t+1, \tasks)$ follows from \Cref{theorem:valid-task-set}.
            
        \end{itemize}
        \item This condition follows directly from the scheduler $DS_{\staticSchedule_a, B}$.
        By taking the task with the earliest deadlines and the partial order $\leq_{DL}$ over the tasks, the task list is never empty.
    \end{enumerate}
\end{proof}

\begin{lemma}
\label{lemma:3}
Given a well-defined specification $\varphi$, a world $\world \in \semantics{\varphi}$, an annotation mapping $a$, and a bound $B = \simpleBound{b}{}$ with $b \ge \max_{\task \in \Tasks} |\task|$.
Then $\world \in \semantics{\dynamicSchedule_{\staticSchedule_a}}$ if $\forall t. \validTasks(\world, t, PS_{\staticSchedule_a,B}(\world, t')))$.
\end{lemma}
\begin{proof}
    By unrolling the defintion of $\world \in \semantics{\dynamicSchedule_{\staticSchedule_a}}$, we get:
     \[
	\forall t \in \time. \forall \task \in \Tasks. \begin{cases}
		\taskSat{\task}{\world}{t + 1} &\text{if } \dynamicSchedule_{\staticSchedule_a}(\world, t)(\task) = \scheduleYes\\
		\top &\text{if } \dynamicSchedule_{\staticSchedule_a}(\world, t)(\task) = \scheduleMaybe\\
		\taskUnSat{\task}{\world}{t + 1} &\text{if } \dynamicSchedule_{\staticSchedule_a}(\world, t)(\task) = \scheduleNo
	\end{cases}.
\] with 

\begin{align*}
	\dynamicSchedule(\world, t)(\task) = \begin{cases}
	\scheduleYes & \text{if } (c_1, p_1) = \lastPrio_{\staticSchedule}(\task, \world, t) \land \exists \task'. (c_2, p_2) = \lastPrio_{\staticSchedule}(\task', \world, t)\\
	& \phantom{if } \land \taskSat{\task'}{\world}{t+1} \land p_1 > p_2 \land \neg\exists \task'' \succeq \task. \taskSat{\task''}{\world}{t+1}\\
	\scheduleMaybe & \text{otherwise}
	\end{cases}
\end{align*}
Given an arbitrary but fix $t \in \time$ 
    and $\task \in \Tasks$ then, we have to proof:
    \[
        \dynamicSchedule_{\staticSchedule_a}(\world,t)(\task) = \scheduleYes \rightarrow \taskSat{\task}{t + 1}{\world}
    \]
    In priority schedules $
        \dynamicSchedule_{\staticSchedule_a}(\world,t)(\task) = \scheduleYes$ is true if there exists a task $\task'$ where the world satisfies this task at the next timestep and this task has a lower priority as $\task$ and $\task'$ is not triggered by a task dependency.
        Since the world follows $PS$ and $PS$ picks the tasks according to their priority, then the world als satisfies this task $\task$ at the next timestamp or $\task = \task_1 \cup \task_2$ and is for this triggered.
        However, due to the construction of $\staticSchedule_a$ the priority of $\task$ is at least the priority of $ \task_1$ and $\task_2$.
        For this, if $\task_1$ and $\task_2$ are in the next taskset so is also $\task$ and the condition holds.
\end{proof}

\begin{theorem}[Valid Priority Scheduler]
    \label{theorem:valid-priority-scheduler}
    Given a well-defined specification $\varphi$, an annotation mapping $a$,  and a bound $B = \simpleBound{b}{}$.
    The scheduler $PS_{\staticSchedule_a,B}$ is valid $b \ge \max_{\task \in \Tasks} |\task|$.
\end{theorem}
\begin{proof}
    By definition, we have to proof the following conditions:
    \begin{enumerate}
		\item
		$\quad\begin{aligned}[t]
			&\forall \world \in \semantics{(\varphi, \dynamicSchedule_{\staticSchedule_a}, B)}. \forall t \in \time. \forall \tasks \subseteq \Tasks. \forall i \in \inputValues.\\
			&\quad \prefix_{PS_{\staticSchedule_a,B}}(\world, t) \land PS_{\staticSchedule_a, B}(\world, t) = \tasks \land i \models \tasks\\
			&\quad\rightarrow \exists \world' \in \semantics{(\varphi, \dynamicSchedule_{\staticSchedule_a}, B)}. \validTasks(\world', t+1, \tasks) \land \world'[..t] = \world[..t] \land \world'[t+1] = i.
		\end{aligned}$
		\item 
		$\quad\begin{aligned}[t]
			&\forall \world \in \semantics{(\varphi, \dynamicSchedule_{\staticSchedule_a}, B)}. \forall t \in \time. \exists t' > t. |PS_{\staticSchedule_a,B}(\world, t')| \ge 1
		\end{aligned}$
	\end{enumerate}
    We proof each condition separately:
    \begin{enumerate}
        \item We proof this condition by constructing $\world'$, with $\world'[..t] = \world[..t]$, $\world'[t+1] = i$, and $\world'$ follows the scheduler $PS_{\staticSchedule_a, B}$. Then we have to proof:
        \begin{itemize}
            \item $\world \in \semantics{\varphi}$ follows from the well-definedness of $\varphi$
            \item $\world \in \semantics{\dynamicSchedule_{\staticSchedule_a}}$ follows from \Cref{lemma:3}
            \item $\forall t \in \time. B(\world, t)$ and $\validTasks(\world', t+1, \tasks)$ follows from \Cref{theorem:valid-task-set}.
            
        \end{itemize}
        \item This condition follows directly from the scheduler $PS_{\staticSchedule_a, B}$.
        By taking the task with the earliest deadlines and the partial order $\leq_{P}$ over the tasks, the task list is never empty.
    \end{enumerate}
\end{proof}

\begin{theorem}[Valid Deadline-Priority Scheduler]
    \label{theorem:valid-deadline-priority-scheduler}
    Given a well-defined specification $\varphi$, an annotation mapping $a$,  and a bound $B = \simpleBound{b}{}$.
    The scheduler $DPS_{\staticSchedule_a,B}$ is valid $b \ge \max_{\task \in \Tasks} |\task|$.
\end{theorem}
\begin{proof}
    The proof follows the same structure as the proof of \Cref{theorem:valid-deadline-priority-scheduler} but adds with the $\mathit{overdue}$ function an additional priority that is higher than the other priorities. This is reflected in the semantics and in the scheduler
\end{proof}